\documentclass[11pt,a4paper]{amsart}
\usepackage[T1]{fontenc}
\usepackage[latin9]{inputenc}
\usepackage{amssymb}
\usepackage{color}
\usepackage{graphicx,color}
\usepackage{amsmath, amssymb, graphics}

\makeatletter
\numberwithin{equation}{section} 
\numberwithin{figure}{section} 
  \@ifundefined{theoremstyle}{\usepackage{amsthm}}{}
  \theoremstyle{plain}
  \newtheorem{thm}{Theorem}[section]
  \theoremstyle{plain}
  
  \theoremstyle{plain}
  
  \theoremstyle{remark}
  \newtheorem{rem}[thm]{Remark}
  \theoremstyle{remark}
  
  \theoremstyle{plain}
  \newtheorem{lem}[thm]{Lemma}

\def\com#1{ \hbox{#1}}

\smallskip
\def\<{{\langle }}
\def\>{{\rangle }}

\def\com#1{ \quad\hbox{#1}\quad}

\smallskip
\def\<{{\langle }}
\def\>{{\rangle }}
\def\k{\rm {\bf Kg}}
\def\m{\rm {\bf m}}
\def\s{\rm {\bf s}}
\def\um{\rm {\bf um}}
\def\ud{\rm {\bf ud}}
\def\ut{\rm {\bf ut}}

\begin{document}

\title[Relativistic Lagrange Points]{Existence and stability of Lagrangian points in the relativistic restricted three body problem}
\author{Oscar Perdomo}
\date{\today}

\curraddr{Department of Mathematics\\
Central Connecticut State University\\
New Britain, CT 06050\\}
\email{ perdomoosm@ccsu.edu}


\begin{abstract}
In this paper we reinvestigate the stability and existence of Lagrangian points in the circular restricted 2+1 body problem treated in the framework of the post-Newtonian approximation of the general relativity. It is well known that the stability of the Lagrangian points in the Newtonian case depends on showing that the real parts of the eigenvalues of a matrix are zero. The reason we are reinvestigating this topic is due to the fact that most of the papers written so far on the stability and existence of the relativistic restricted three body problem are not mathematically correct, they have one of the following well known mathematical errors: {\bf 1.} Showing that an expression is close to a small positive number does not show that this expression is positive {\bf 2.} Showing that the approximation of an expression is zero does not show that the expression is zero and {\bf 3.} Showing one solution of a system of equations that is obtained by doing a small perturbation of another system of equations, does not show the existence of a solution of the latter system of equations. 

In the Newtonian case the parameter $\mu=\frac{m_2}{m_1+m_2}$  describes, up to symmetries, all the possible restricted circular 2+1 body problem. Here $m_1$ and $m_2$ are the masses of the primaries. In the relativistic case, we need two parameters to describe all the possible systems, one is $\mu$ and the other one is the number $c$ that represents the speed of light in the units where the period  of the primaries is $2 \pi$ and the distance between the primaries is 1. We point out that $c$ is not necessarily a big number, it is about 10065 for the Sun-Earth system, 6262 for the Sun-Mercury system and it is about 683 for the Pulsar Binary star system, under the assumption that the Pulsar binary star system is circular. 

Even though it seems to be almost impossible to find a closed form for the Lagrangian points in the relativistic three body problem, we show in this paper how the Poincare-Miranda theorem can be used to prove the existence of Lagrangian points. One of the main results in this paper provides the exact expression for the characteristic polynomial of the matrix that determines the stability of the Lagrange points. We point out that even without having a closed form for the Lagrangian points, we can show that the characteristic polynomial has the form $\lambda^4+a_1 \lambda^2+ a_2$, with  the expression for $a_1$ and $a_2$ depending on $\mu$, $c$ and the coordinates of the Lagrangian point. The form of this polynomial shows that we indeed have stability results similar to those shown in the Newtonian case. At the end of the paper we find the coordinates of Lagrangian points for some particular systems with a precision smaller than $10^{-30}$ and we compare the new results with those already found in the literature. We conclude that the error in the previous results is big.
\end{abstract}

\maketitle

\section{Introduction}

Lagrangian points could be seen as particular periodic solutions of the restricted three body problem and  they are extensively used in space missions. Even thought for practical reasons  the effect of relativity theory in the computation of the Lagrangian points may be irrelevant to a mission, it is a good idea to know the exact position of the Lagrangian points when relativity theory is considered and their stability in order to make decisions whether or not working with Newtonian physics is good enough. With the intension to set up notation for the change of units and to get familiar with the terminology, section 2 explains the Newtonian case. Section 3 displays the ODE for the relativistic case. We point out that it is not known that there are also exactly five equilibrium solutions of the ODE in the relativistic case, it is natural to conjecture that we also have 5 of these points and it is natural to also  call them Lagrangian points. We point out that the existence is far from being obvious -especially when $\mu$ is small-  due to the fact that the case $\mu=0$ degenerates to the case where there is only one massive body and two massless bodies going around. If we were to assign equilibrium points in this case, we would have to say that the ODE has infinitely many equilibrium points, all of them forming a circle. In this way the Newtonian restricted three body problem moves from an ODE having infinitely many equilibrium points when $\mu=0$ to an ODE having only five equilibrium points when $\mu>0$. In section \ref{ex} we will prove the existence of the point $L_4$ for the Earth-Sun system using the Poincare-Miranda Theorem. We point out that the existence of other equilibrium points can be done in a similar way for any other system associated with other parameters  $\mu$ and $c$. The existence of the equilibrium points $L_1$, $L_2$ and $L_3$ is somehow easier because it does not require the Poincare-Miranda theorem, in this case the problem reduces to solve an equation with one variable and therefore the intermediate value theorem can be applied. Regardless of how many equilibrium points we have, or where they are located, section \ref{cps} shows that the characteristic polynomial at {\rm any} of the equilibrium solutions -Lagrangian points-  has the form $\lambda^4+a_1 \lambda^2+ a_2$. As a consequence we obtain that whenever the two roots of the quadratic equation $\sigma^2+a_1 \sigma+ a_2=0$ are negative, then,  the roots of $\lambda^4+a_1 \lambda^2+ a_2$ have zero real part and therefore the equilibrium point is linearly stable. At this point we would like to point out that proving that a quantity is zero  cannot be done by considering approximations. For this reason, for the case of $L_4$, it is not surprising to have papers like Bhatnagar and Hallan \cite{BH} where they conclude that $L_4$ is unstable because the real part of some eigenvalues are close to a positive very small number in contrast with papers like Douskos et al \cite{D} and Ahmen et al \cite{A} where they conclude that  $L_4$ is stable for some values of the parameter $\mu$  because after a set of rounding of order $\frac{1}{c^2}$ and $\frac{1}{c^3}$ they obtain that the real part of some eigenvalues is zero. So, which one of these results is true having in mind that all of them have used rounding? We will show, even without having a closed form for the equilibrium points, that $L_4$ is stable for some open region on the set of parameters $\mu$ and $c$.

 On section \ref{comp} we consider the relativistic restricted three body problem coming from the choice of parameters $\mu=0.0384$ and  $c$  $=$ 4, 10, 50, 100, 400, 800, 1600, 3200, 6400, 12800. We use this 10 ODE systems  to compare the previous results with those obtained in this paper. We conclude that in all of these systems, the rounding error in the papers \cite{D} and \cite{A} is big. We would like to point out that according to some authors, for example \cite{M}, papers \cite{A} and \cite{D} are considered to be among the latest results regarding the stability of the Lagrangian points for the relativistic restricted three body problem.
 
The author would like to thank David R. Skillman and Andr\'es Mauricio Rivera for his valuable comments.

\section{Circular solutions of the two body problem and changing units} \label{units} Let us consider two bodies (the primaries) with masses $m_1\,  \k$ and $m_2\,  \k$ which moves in the space with positions $x$ and $y$. Let us take the gravitational constant to be equal to $G=6.67384*10^{-11}\, \m^3\, \k^{-1}\, \s^{-2}$. It is easy to check that for a given positive number $a$ the functions

$$
x(t)=\frac{-m_2\, a}{m_1+m_2} \, (\cos(\omega\,  t),\sin(\omega \, t)), \quad y(t)=\frac{m_1\, a}{m_1+m_2} \, (\cos(\omega\,  t),\sin(\omega \, t))
$$

with $\omega=\sqrt{\frac{G(m_1+m_2)}{a^3}}\, {\bf s}^{-1}$ satisfy the two body problem ODE 

$$\ddot{x}=\frac{m_2 G}{|x-y|^3} \, (y-x)\quad \ddot{y}=\frac{m_1 G}{|y-x|^3} \, (x-y)$$

This solution satisfies that the distance between the two bodies is always $a$ meters and moreover, both motions are periodic since they complete a revolution after $T=\frac{2 \pi}{\omega}=2 \pi\, \sqrt{\frac{a^3}{G(m_1+m_2)}}\, {\bf s} $. Let us change the units of mass, distance and time in the following way: Let us denote by $\um$ the unit of mass such that $1 {\um} = (m_1+m_2)\,  \k$, let us denote by $\ud$ the unit of distance such that $1\,  {\ud}=a\,  \m$ and  finally let us denote by $\ut$,  the unit of time $\ut$ such that 
$1\, {\ut} = \sqrt{\frac{a^3}{G (m_1+m_2)}}\quad \s$. Notice that using the units $\ut$ and $\ud$ we have that the distance between the two bodies is $1 \ud$ and the period of the motion is $2 \pi \, \ut$. We also have that the gravitation constant is $1\, 
\ud^3\, \um^{-1}\, \ut^{-2}$. We point out that the speed of light is

\begin{eqnarray}\label{sl}
 c=299792458 * \sqrt{\frac{a}{G (m_1+m_2)}}\, \frac{\ud}{\ut}
 \end{eqnarray}

If we denote by  $\mu=\frac{m_2}{m_1+m_2}$ and we work in the new units $\ut$, $\um$, $\ud$, then the mass of the first and  second body are $1-\mu$ and $\mu$ and  the motion of the primaries are given by 

$$ x(t)=-\mu \left(\cos (t),\sin (t)\right) \quad  y(t)=(1-\mu ) \, \left(\cos (t),\sin (t)\right) $$

Moreover, if a third body with position $z(t)$ and neglecting mass compare with $m_1$ and $m_2$ moves under the influence of the gravitational force of the primaries, then $z$ satisfies

\begin{eqnarray}\label{eq1}
\ddot{z}=\frac{(1-\mu)}{|x-z|^3}\, (x-z)+\frac{\mu}{|y-z|^3}\, (y-z)
\end{eqnarray}

A direct computation shows that if we take 

\begin{eqnarray}\label{eq2}
z = \left(\xi (t) \cos (t)-\eta (t) \sin (t),\eta (t) \cos (t)+\xi (t) \sin (t)\right)\, ,
\end{eqnarray}

then, (\ref{eq1}) reduces to 
\begin{eqnarray}\label{eq3}
\ddot{\xi}-2 \dot{\eta}=\frac{\partial w_0}{\partial \xi}
\com{and} 
\ddot{\eta}+2 \dot{\xi}=\frac{\partial w_0}{\partial \eta}
\end{eqnarray}

where,

\begin{eqnarray}\label{eq4}
w_0 = \frac{1}{2} (\xi^2+\eta^2)+\frac{1-\mu}{\sqrt{(\xi+\mu)^2+\eta^2}}+\frac{\mu}{\sqrt{(\xi+\mu-1)^2+\eta^2}}
\end{eqnarray} 

A direct verification shows that $\xi(t)=\frac{1-2 \mu}{2}$ and $\eta(t)=\frac{\sqrt{3}}{2}$ is a solutions of the (\ref{eq3}). This equilibrium point $(\frac{1-2 \mu}{2},\frac{\sqrt{3}}{2})$ is known as the Lagrangian point $L_4$. In order to analyze the stability of $L_4$ we consider the function 

$$F_0=(\dot{\xi} ,2 \dot{\eta}+\frac{\partial w_0}{\partial \xi },\dot{\eta} ,\frac{\partial w_0}{\partial \eta }-2 \dot{\xi} ) $$

as a function of the variables $\phi=(\xi,\dot{\xi},\eta,\dot{\eta})$. It is easy to check that the ODE (\ref{eq3}) is equivalent to the ODE
$\dot{\phi}=F_0(\phi)$. In order to analyze the stability of the equilibrium solution $\phi_0=(\frac{1-2 \mu}{2},0,\frac{\sqrt{3}}{2},0)$, we compute the 4 by 4 matrix $A_0=DF_0$ evaluated at $  \xi=\frac{1-2 \mu}{2},\, \eta=\frac{\sqrt{3}}{2},\, \dot{\xi}=0, \, \dot{\eta}=0$. Since we can check that the characteristic polynomial of the matrix $A_0$ is equal to

$$\lambda ^4+\lambda ^2-\frac{27}{4} (\mu -1) \mu$$

Then, we conclude that, when either $0<\mu<\frac{1}{18} \left(9-\sqrt{69}\right)$ or $\frac{1}{18} \left(9+\sqrt{69}\right)<\mu<1$, then the real part of all the eigenvalues of $A_0$ is zero and therefore $L_4$ is linearly stable. It is known that there are 5 equilibrium solutions for the ODE (\ref{eq3}); we have $L_4$, given above, $L_5=(\frac{1-2 \mu}{2},-\frac{\sqrt{3}}{2})$ and three more of the form $(\xi_1,0)$, $(\xi_2,0)$ and $(\xi_3,0)$ usually label as the Lagrangian points $L_1$, $L_2$ and $L_3$. A similar analysis to the one that we just did for $L_4$ can be done for the other Lagrange points to conclude that $L_5$ is also linear stable for the same range of the parameter $\mu$ and, $L_1$, $L_2$ and $L_3$ are linearly unstable.

\section{The ODE in the relativistic case:} The equation of the motion for the restricted three body problem are very similar to the one given by Equation (\ref{eq3}), it takes the form (see Brumberg, 1972, \cite{B} and Bhatnagar \cite{BH})

\begin{eqnarray}\label{eq6}
\ddot{\xi}-2 n \dot{\eta}=\frac{\partial w}{\partial \xi}-\frac{d}{dt} (\frac{\partial w}{\partial \dot{\xi}})\com{and} 
\ddot{\eta}+2 n \dot{\xi}=\frac{\partial w}{\partial \eta}-\frac{d}{dt} (\frac{\partial w}{\partial \dot{\eta}})
\end{eqnarray}
 
 where $w=w_0+\frac{1}{c^2}\, w_1$ with
 
\begin{eqnarray*}
w_1 &=&-\frac{3}{2} \left(1-\frac{1}{3} \mu  (1-\mu )\right) {\rho}^2+\frac{1}{8} \left(\dot{\eta}^2+2 (\dot{\eta}\xi -\dot{\xi}\eta )+\dot{\xi}^2+{\rho}^2\right)^2+\\
& & \frac{3}{2} \left(\frac{1-\mu }{{\rho_1}}+\frac{\mu }{{\rho_2}}\right) \left(\dot{\eta}^2+2 (\dot{\eta}\xi -\dot{\xi}\eta )+\dot{\xi}^2+{\rho}^2\right)-\frac{1}{2} \left(\frac{(1-\mu )^2}{{\rho_1}^2}+\frac{\mu ^2}{{\rho_2}^2}\right)+\\
& & \mu (1-\mu)\left(\left(4 \dot{\eta}+\frac{7 \xi }{2}\right) \left(\frac{1}{{\rho_1}}-\frac{1}{{\rho_2}}\right)-\frac{1}{2} \eta ^2 \left(\frac{\mu }{{\rho_1}^3}+\frac{1-\mu }{{\rho_2}^3}\right)+\left(\frac{3 \mu -2}{2 {\rho_1}}-\frac{1}{{\rho_1} {\rho_2}}+\frac{1-3 \mu }{2 {\rho_2}}\right)\right.\\
n&=&1-\frac{3}{2 c^2} \left(1-\frac{1}{3} \mu(1-\mu) \right)\\
\rho&=& \sqrt{\xi^2+\eta^2},\quad \rho_1=\sqrt{(\xi+\eta)^2+\eta^2}\com{and} \rho_2=\sqrt{(\xi+\eta-1)^2+\eta^2}
\end{eqnarray*}

\subsection{The system of equations}

The equilibrium points of the system of differential equations given by (\ref{eq6}) are the solutions of the system $f=0$ and $g=0$ where  $f=\frac{\partial w}{\partial \xi }$ and $g=\frac{\partial w}{\partial \eta }$ evaluate at $\dot{\xi}=\dot{\eta}=0$.

\section{Existence  of $L_4$ in the relativistic case using the Poincare-Miranda theorem}\label{ex}

\subsection{Existence of $L_4$ for the Earth-Sun system}

For computation in this section, we will take earth and the sun moving with a constant distance between them of $a_0=149597870700$ {\bf m} with  the mass of the sun equal to $M_0=1.988544* 10^{30}\, {\bf Kg}$ and the mass of the Earth equal to $5.9729 *10^{24} {\bf Kg}$. We will also will be taking the speed of light to be $c_0=299792458 \, \frac{{\bf m}}{{\bf s}}$. Using this data we have that the values for $\mu$ and $c$ are given by 
$$\mu= \frac{59729}{19885499729}\approx 3.00365*10^{-6}$$

and

$$ c=c_0 * \sqrt{\frac{a_0}{G (m_0+M_0)}}=\frac{149896229 \sqrt{\frac{1495978707}{3317816087784734}}}{10}\approx 10065.3$$

We will prove the existence of the relativistic point $L_4$ for the system Earth-Sun-mass zero body. This is, we will prove the existence of a point $(\xi_0,\eta_0)$ that is within a distance of $10^{-15}$ of the point $(\frac{1-2 \mu}{2},\frac{\sqrt{3}}{2})$ that satisfies the equation 
\begin{eqnarray}\label{the eq}
f(\xi_0,\eta_0)=0\com{and} g(\xi_0,\eta_0)=0
\end{eqnarray} 

In order to prove the existence of $(\xi_0,\eta_0)$, let us consider the following five points

\begin{eqnarray*}
Z_0&=&\left(\frac{2499985012616009587660193140271}{5000000000000000000000000000000},\frac{1082531750278361975463116188557}{1250000000000000000000000000000}\right)\\
P_1&=& \left(\frac{2499985512616009587660193140271}{5000000000000000000000000000000},\frac{4330127145451017576343432330693}{5000000000000000000000000000000}\right)\\
P_2&=&\left(\frac{2499984762616009587660193140271}{5000000000000000000000000000000},\frac{4330127145451017576343432330693}{5000000000000000000000000000000}\right) \\
P_3&=&\left(\frac{2499985512616009587660193140271}{5000000000000000000000000000000},\frac{2165063356219154276435264800649}{2500000000000000000000000000000}\right) \\
P_4&=&\left(\frac{2499984762616009587660193140271}{5000000000000000000000000000000},\frac{2165063356219154276435264800649}{2500000000000000000000000000000}\right)
\end{eqnarray*}

And let $\beta_1$ be the line that connects $P_2$ with $P_1$, $\beta_2$ be the line that connect $P_4$ with $P_3$, $\beta_3$ be the line that connect $P_3$ with $P_1$ and $\beta_4$ be the line that connect $P_4$ with $P_2$. More precisely,
\begin{eqnarray*}
\beta_1(t)&=&t P_1+(1-t) P_2\\
\beta_2(t)&=&t P_3+(1-t) P_4\\
\beta_3(t)&=&t P_1+(1-t) P_3\\
\beta_4(t)&=&t P_2+(1-t) P_4\\
\end{eqnarray*}

\begin{figure}[hbtp]
\begin{center}\includegraphics[width=.7\textwidth]{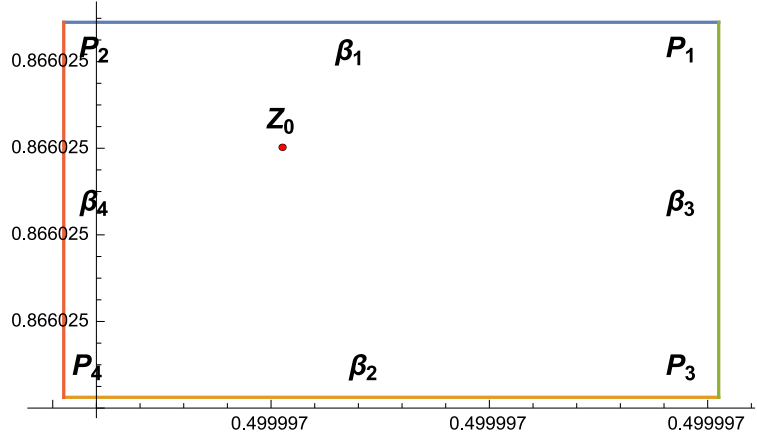}
\end{center}
\caption{The Poincare Miranda theorem guarantees that the system of equations $f=0$ and $g=0$ has a solution inside the region above }\label{fig1}
\end{figure}
\begin{thm}
There is solution of the system of equations  $f=0$ and $g=0$ inside the region delimited by the union of the curves $\beta_1$, $\beta_2$, $\beta_3$ and $\beta_4$. Moreover, we have that the first six significant digits of the functions evaluated at the points $Z_0$, $P_1$, $P_2$, $P_3$ and $P_4$ are given by 

\begin{eqnarray*} 
f(P_1)&=& 1.124997\dots * 10^{-7} \qquad g(P_1)\, =\,  1.94854\dots * 10^{-7}\\
f(P_2)&=& -2.22772\dots * 10^{-13} \qquad g(P_2)\, =\,  3.94516\dots * 10^{-13}\\
f(P_3)&=& 4.60545\dots * 10^{-13}  \qquad g(P_3)\, =\, -7.63052\dots * 10^{-13}\\
f(P_4)&=& -1.12499\dots * 10^{-7} \qquad g(P_4)\, =\, - 1.94855\dots * 10^{-7} \\
f(Z_0)&=& -1.14508 \dots * 10^{-32} \qquad g(Z_0)\, =\, -4.25190 \dots * 10^{-32} 
\end{eqnarray*}

We also have that the function $g>0$ on $\beta_1$, $g<0$ on $\beta_2$, $f>0$ on $\beta_3$ and $f<0$ on $\beta_4$. As a consequence of the Poincare-Miranda Theorem we conclude that there exists a point $P_0=( \xi_0,\eta_0)$ inside the region bounded by the four curves $\beta_i$ such that  $f(P_0)=g(P_0)=0$.
\end{thm}

\begin{proof}
The proof to the theorem relies on the fact that we have an exact expression (an analytic expression) for $f$ and $g$ and the fact that programs like Mathematica allow us to precisely compute a certain amount of real digits of an exact expression. In order to obtain these digits, it is required that we work with exact numbers, that is the reason we decided to use rational numbers and not decimals. The computation for the values of the functions $f$ and $g$ were obtained by using the command $\text{RealDigits}[f(Z_0),10,6]$ and $\text{RealDigits}[g(Z_0),10,6]$ form the program Wolfram Mathematica 10. Likewise for the points $f(P_i)$ and $g(P_i)$. It is not difficult to check that the directional derivative of $g$ along the velocity of the curves $\beta_1$ and $\beta_2$ does not change sign, they both are approximately $1.29903$, and also, the directional derivative of $f$ along the curves $\beta_3$ and $\beta_4$ does not change sign, they both are  approximately $1.29903$. Then, we conclude that the function $g$ is monotonic along $\beta_1$ and $\beta_2$, this fact along with the values of $g$ at the endpoints allows to prove that $g$ is positive on $\beta_1$ and negative on $\beta_2$. A similar arguments holds for the function $f$. This concludes the proof  of the  theorem.
\end{proof}

\begin{rem}
The value $Z_0$ was initially computed to have a good approximation of the equations $f=0$ and $g=0$ near $L_4$. This good approximation is needed to find the curves $\beta_i$ that satisfy the conditions of the Poincare-Miranda Theorem.
\end{rem}

\section{Exact expression for the characteristic polynomial at the equilibrium points }\label{cps}

The following theorem provides an expression for the characteristic polynomial of a system of the type given by the relativistic three body problem.

\begin{thm}\label{cp}
Let us consider the potential $U=U(x,\dot{x},y,\dot{y})$ and let us consider the following system of ODE

$$\ddot{x}-2 k \dot{y}=\frac{\partial U}{\partial x}-\frac{d\, }{dt}(\frac{\partial U}{\partial \dot{x}}),  \quad
\ddot{y}+2 k \dot{x}=\frac{\partial U}{\partial y}-\frac{d\, }{dt}(\frac{\partial U}{\partial \dot{y}})$$
\end{thm}
where $k$ is a constant. If $L_0=(x_0,y_0)$ is an equilibrium point of the system above and 
$d= 1+\frac{\partial^2U}{\partial \dot{x}^2}+\frac{\partial^2U}{\partial \dot{y}^2}+\frac{\partial^2U}{\partial \dot{x}^2} \frac{\partial^2U}{\partial \dot{y}^2}-\left(\frac{\partial^2U}{\partial \dot{x} \partial \dot{y}}\right)^2$ is not zero at $\tilde{L}_0=(x_0,0,y_0,0)$, then, the characteristic polynomial of the matrix that describes the linearization of the ODE at $L_0$ is given by

$$ \lambda^4 + a_1\lambda^2+a_2$$

where,
\begin{eqnarray*}
a_2\, d&=&\frac{\partial^2U}{\partial y^2} \frac{\partial^2U}{\partial x^2}-\left(\frac{\partial^2U}{\partial x \partial y}\right)^2
\end{eqnarray*}

and 

\begin{eqnarray*}
a_1\, d&=&-4 k \frac{\partial^2U}{\partial y \partial \dot{x}}+4 k \frac{\partial^2U}{\partial x \partial \dot{y}}+\left(\frac{\partial^2U}{\partial y \partial \dot{x}}\right)^2+\left(\frac{\partial^2U}{\partial x \partial \dot{y}}\right)^2-\frac{\partial^2U}{\partial y^2}-\frac{\partial^2U}{\partial y^2} \frac{\partial^2U}{\partial \dot{x}^2}-\\
& &2 \frac{\partial^2U}{\partial y \partial \dot{x}} \frac{\partial^2U}{\partial x \partial \dot{y}}+2 \frac{\partial^2U}{\partial \dot{x} \partial \dot{y}} \frac{\partial^2U}{\partial x \partial y}-\frac{\partial^2U}{\partial \dot{y}^2} \frac{\partial^2U}{\partial x^2}-\frac{\partial^2U}{\partial x^2}+4 k^2
\end{eqnarray*}

\begin{proof}
The ODE in this theorem can be reduce to the first order ODE $\dot{\phi}=F(\phi)$ with $\phi=(x,\dot{x},y,\dot{y})$ and 
$$ F(\phi)=(\dot{x},{F_2}(\phi),\dot{y},{F_4}(\phi))$$
and the functions ${F_2}$ and ${F_4}$ are given as the solution, near $\tilde{L}_0=(x_0,0,y_0,0)$, of the system of equations

\begin{eqnarray*}
{F_2}-2k\dot{y}&=&\frac{\partial U}{\partial x}-\dot{x} \frac{\partial^2U}{\partial \dot{x} \partial x}-{F_2} \frac{\partial^2U}{\partial \dot{x}^2}-\dot{y} \frac{\partial^2U}{\partial \dot{x} \partial y}-{F_4} \frac{\partial^2U}{\partial \dot{x} \partial \dot{y}}\\
{F_4}+2k\dot{x}&=&\frac{\partial U}{\partial y}-\dot{x} \frac{\partial^2U}{\partial \dot{y} \partial x}-{F_2} \frac{\partial^2U}{\partial \dot{x} \partial \dot{y} }-\dot{y} \frac{\partial^2U}{\partial \dot{y} \partial y}-{F_4} \frac{\partial^2U}{\partial \dot{y}^2}
\end{eqnarray*}

Recall that we have that ${F_2}(\tilde{L}_0)={F_4}(\tilde{L}_0)=0$. If we compute the partial derivative with respect to $x$ to the system of equations above and we evaluate at $\tilde{L}_0$, we get the following system of equation 

\begin{eqnarray*}
{F_2}_x&=&\frac{\partial^2 U}{\partial x^2}-{F_2}_x \frac{\partial^2U}{\partial \dot{x}^2}-{F_4}_x \frac{\partial^2U}{\partial \dot{x} \partial \dot{y}}\\
{F_4}_x&=&\frac{\partial^2 U}{\partial x\partial y}-{F_2}_x \frac{\partial^2U}{\partial \dot{x} \partial \dot{y} }-{F_4}_x \frac{\partial^2U}{\partial \dot{y}^2}
\end{eqnarray*}

This is a linear system on ${F_2}_x$ and  ${F_4}_x$ with solution solution satisfying,

\begin{eqnarray*}
{F_2}_x\, d&=&\frac{\partial^2U}{\partial \dot{y}^2} \frac{\partial^2U}{\partial x^2}+\frac{\partial^2U}{\partial x^2}-\frac{\partial^2U}{\partial \dot{x} \partial \dot{y}} \frac{\partial^2U}{\partial x \partial y}\\
{F_4}_x\, d&=& \frac{\partial^2U}{\partial \dot{x}^2} \frac{\partial^2U}{\partial x \partial y}+\frac{\partial^2U}{\partial x \partial y}    -
\frac{\partial^2U}{\partial \dot{x} \partial \dot{y}} \frac{\partial^2U}{\partial x^2}
\end{eqnarray*}

Likewise we can obtain expression for ${F_2}_y$ and ${F_4}_y$ and for ${F_2}_{\dot{x}}$, $g_{\dot{x}}$ and finally for ${F_2}_{\dot{y}}$, ${F_4}_{\dot{y}}$ evaluated at the point $\tilde{L}_0$. The theorem follow after replacing these expression for the partial derivative of the functions ${F_2}$ and ${F_4}$ into the characteristic polynomial of the matrix

$$
\left(
\begin{array}{cccc}
 0 & 1 & 0 & 0 \\
 {F_2}_x & {F_2}_{\dot{x}} & {F_2}_y & {F_2}_{\dot{y}} \\
 0 & 0 & 0 & 1 \\
 {F_4}_x & {F_4}_{\dot{x}} & {F_4}_y & {F_4}_{\dot{y}} \\
\end{array}
\right)
$$

\end{proof}

\section{Comparing the results in this paper with those found before}\label{comp}

As we pointed out before, the value $c$ that represents the speed of light with respect to the units {\bf ut} and {\ud}, varies according to the formula $ c=c_0 * \sqrt{\frac{a}{G M}}$ where $c_0=299792458$ and $M=m_1+m_2$ is the mass of the system. Since the value $c$ varies, it is part of the parameters that describe the relativistic restricted three body problem. 

If we want to talk about error, we would like to go back to the units meters and seconds. Recall that $a$ represent the distance between the primaries in meters and $M$ is the mass of the system in Kilograms. In this section we will be comparing our results with those obtained in 2002 by Douskos and Perdios, \cite{D}, and the results obtained in 2006 by Ahmed, El-Salam and El-Bar \cite{A}. 

Notice that when the equilibrium point $L_4$ is stable, then the roots of the characteristic polynomial are of the form $\pm \omega_1\, i$ and $\pm \omega_2\, i$. In this case, we expect to have periodic solutions (in the synodic frame of reference $(\xi,\eta)$) with periods close to 

\begin{eqnarray}\label{t1t2}
T_1= \frac{2 \pi}{\omega_1}\com{and} T_2=\frac{2 \pi}{\omega_2}
\end{eqnarray}

Since in \cite{A} the authors show the stability of $L_4$ for all positive values of $c$ as long as $0<\mu<0.0384$ and in \cite{D} the authors show the stability of $L_4$ for all values of $c>\sqrt{\frac{65}{6}}\approx 3.291$ as long as $0<\mu<\frac{1}{2}-\frac{\sqrt{69}}{18}\left( 1+\frac{17}{27 c^2}\right)$, then, in order to compare the error of the previous two papers, we will consider 10 systems, each one of them with total mass equal to the mass of the sun $M_0$,  $\mu=0.034$, and with $c$ taking the values $c_1=4$, $c_2=10$, $c_3=50$, $c_4=100$, $c_5=400$, $c_6=800$, $c_7=1600$, $c_8=3200$, $c_{9}=6400$ and $c_{10}=12800$. For each one of these 10 points we will compute the coordinates $Z_i=(\xi_i,\eta_i)$ of a solution near  $Z_0=\left(\frac{1-2 \mu }{2} (\frac{5}{4 c^2}+1) ,\frac{\sqrt{3}}{2}\, (1-\frac{6 \mu ^2-6 \mu +5}{12 c^2})\right)$ that solve the system of equations $f=g=0$ with a precision smaller than $10^{-30}$ and we will call this solution the ``exact solution''. We will compute $|Z_i-Z_0|$ for each of the 10 points and the order of precision of the solution $Z_0$. We will also compute, using the exact solution,  the values for $T_1$ and $T_2$ defined at the beginning in this section and then we will compare them with  the values for $T_1$ and $T_2$ found in the papers \cite{A} and \cite{D}.  We need the following two expression for the characteristic polynomial in order to provide the values of $T_1$ and $T_2$ from the previous papers.
\begin{rem}
Douskos and Perdios, \cite{D}, found that for values of $c$ bigger than $\sqrt{\frac{65}{6}}$, the  characteristic equation of the ODE (\ref{eq6}) can be approximated by the polynomial

$$ p_{2002}=\lambda^4 +\left(1-\frac{9}{c^2}\right)\lambda^2+\frac{9\mu(1-\mu)}{4}\left(3-\frac{65-12\mu(1-\mu)}{2 c^2} \right)$$

 Ahmed, El-Salam and El-Bar, \cite{A}, found that the  characteristic equation of the ODE (\ref{eq6}) can be approximated by the polynomial

$$ p_{2006}=\lambda^4 +\left(\frac{\frac{33}{4}-\frac{189}{16} \mu  (1-\mu )}{c^2}+1\right)\lambda^2+\frac{405 \mu ^4}{32 c^2}-\frac{405 \mu ^3}{16c^2}+\frac{10521 \mu ^2}{256 c^2}-\frac{7281 \mu }{256 c^2}+\frac{27}{4} (1-\mu ) \mu$$

\end{rem}

Before we proceed we point out that the distance between the primaries  $a$ and their period $T$ is related to $c$ and $M$ in the following way

\begin{eqnarray}
a=GM\left(\frac{c}{c_0} \right)^2 \com{and} T=2 \pi M G\left(\frac{c}{c_0}\right)^3
\end{eqnarray}

From the equation above, we easily conclude the following lemma,
 
\begin{lem}\label{error}  An error $\Delta d\, {\ud}$ in distance  and an error $\Delta t \, {\ut}$ in time in the synodic frame of reference $(\xi,\eta)$ correspond to an error of 

$$ a\, \Delta d= GM\left(\frac{c}{c_0} \right)^2\, \Delta d\quad {\bf meters}$$

and 

$$ \frac{T}{2\pi}\, \Delta t=  M G\left(\frac{c}{c_0}\right)^3\, \Delta t\quad {\bf seconds}$$
 
\end{lem}

The following table provide solution of the equation $f=0$ and $g=0$ with a precision smaller than $10^{-30}$, this is, we have that 

$$ |f(\xi_i,\eta_i)|<10^{-30}\com{and} |g(\xi_i,\eta_i)|<10^{-30}  $$
\begin{center}
    \begin{tabular}{| c | c| }
 \hline\noalign{\medskip}
 $c_i$ & solution $(\xi_i,\eta_i)$ with precision $<\, 10^{-30}$ \\
 \hline\noalign{\medskip}
  $4$ &$\displaystyle{ \left(\frac{1269274068083047668315001319947}{2500000000000000000000000000000},\frac{2099727919061389308673386312351}{2500000000000000000000000000000}\right) }$\\
\hline  \noalign{\medskip}
$10$ &$\displaystyle{ \left\{\frac{589933273547627837960417751707}{1250000000000000000000000000000},\frac{431230420634190356869315441943}{500000000000000000000000000000}\right\}} $\\
\hline\noalign{\medskip}
$50$ & $\displaystyle{\left\{\frac{4662331909210469007263660596223}{10000000000000000000000000000000},\frac{4329433007965962475682519470747}{5000000000000000000000000000000}\right\}}$ \\
\hline\noalign{\medskip}
$100$ & $\displaystyle{\left\{\frac{145643206851728280439111229549}{312500000000000000000000000000},\frac{4329953660006884115357445313463}{5000000000000000000000000000000}\right\} }$\\ 
\hline\noalign{\medskip}
$ 400 $ & $\displaystyle{\left\{\frac{291252275419734701378298811871}{625000000000000000000000000000},\frac{8660232373592265679769530789291}{10000000000000000000000000000000}\right\}}$\\
\hline\noalign{\medskip}
$800$ & $\displaystyle{\left\{\frac{932001820318321886652316804353}{2000000000000000000000000000000},\frac{8660248621851491754868337036919}{10000000000000000000000000000000}\right\} }$\\
\hline\noalign{\medskip}
$ 1600$ & $\displaystyle{\left\{\frac{4660002275392444335389570820631}{10000000000000000000000000000000},\frac{4330126341925273094801840139691}{5000000000000000000000000000000}\right\}} $\\
\hline\noalign{\medskip}
$3200$ & $\displaystyle{\left\{\frac{4660000568847769958396390882401}{10000000000000000000000000000000},\frac{8660253699346200359582469403071}{10000000000000000000000000000000}\right\}} $\\
\hline\noalign{\medskip}
$ 6400$ &$ \displaystyle{\left\{\frac{7456000227539073870838002657}{16000000000000000000000000000},\frac{1082531744152482132899736029813}{1250000000000000000000000000000}\right\}} $\\
 \hline\noalign{\medskip}
 $12800$ &$\displaystyle{ \left\{\frac{2330000017776489479899171749009}{5000000000000000000000000000000},\frac{8660254016688255186688034652061}{10000000000000000000000000000000}\right\}} $\\
 \hline\noalign{\medskip}
\end{tabular}
\end{center}

The precision of the solutions $Z_0=\left(\frac{1-2 \mu }{2} (\frac{5}{4 c^2}+1) ,\frac{\sqrt{3}}{2}\, (1-\frac{6 \mu ^2-6 \mu +5}{12 c^2})\right)$ are given by the following two table

\begin{center}
    \begin{tabular}{| c | c|c|c|c|c| }
    \hline
    $c_i$ & $4$&$10$&$50$&$100$&$400$\\
    \hline\noalign{\medskip}
    Max$ \{ |f(Z_0)|,|g(Z_0)| \}$ &$0.0025$&$0.0000701$&$1.13*10^{-7}$&$7.1*10^{-9}$&$2.7*10^{-11}$\\
    \hline\noalign{\medskip}
    \end{tabular}
 \end{center}
 
 \begin{center}
    \begin{tabular}{| c | c|c|c|c|c| }
    \hline
    $c_i$ & $800$&$1600$&$3200$&$6400$&$12800$\\
    \hline\noalign{\medskip}
    Max$ \{ |f(Z_0)|,|g(Z_0)| \}$ &$1.7*10^{-12}$&$1.08*10^{-13}$&$6.7*10^{-15}$&$4.23773*10^{-16}$&$2.6*10^{-17}$ \\
    \hline\noalign{\medskip}
    \end{tabular}
 \end{center}
 
 \begin{rem}
 A direct verification shows that  if $c=4$ and $\mu=0.034$ then, using theorem \ref{cp} to compute the characteristic polynomial, we obtain that its roots are  the four values given by $\pm 0.0878256\pm0.580403 i$. Therefore this equilibrium point is not stable. This results contradicts the theorem shown in \cite{A} where they proved that this equilibrium point must be stable.
 
The following table compares the roots of the following polynomials:

\begin{itemize}
\item
 The polynomial provided by Theorem \ref{cp} after replacing $(\xi,\eta)$ with the solutions with a precision smaller than $10^{-30}$ given above.
 \item
  The polynomial obtained in the 2002 paper \cite{D}. 
  \item  
  The polynomial obtained in the 2006 paper \cite{A}.
\item
The Newtonian polynomial, this is, the polynomial $\lambda ^4+\lambda ^2-\frac{27}{4} (\mu -1) \mu$. This is the polynomial that we obtained when we do not use relativity. In particular we have replaced $(\xi,\eta)$ with$\left(\frac{1-2 \mu }{2}  ,\frac{\sqrt{3}}{2} \right)$ 
 \end{itemize}
 
 \end{rem}
 {\tiny
\begin{center}
    \begin{tabular}{| c | c| c | c | c |}
 \hline\noalign{\medskip}
 $c_i$ & $p_{\, \rm exact}=0 $& $p_{2002}=0 $& $p_{2006}=0 $ & $p_{Newton}=0 $\\
 \hline\noalign{\medskip}
  $4$ &$\pm 0.0878256 \pm 0.580403 i$&$\pm 0.218784 i,\pm1.2307 i$&$\pm 0.345951 i, \pm1.17119 i$&$\pm 0.5759905  i, \pm 0.817456 i $\\
\hline  \noalign{\medskip}
$10$ &$\pm 0.594508336, \pm 0.751015 i$&$\pm 0.479625 i, \pm0.92734 i$&$\pm 0.50929 i,\pm 0.905121 i$&$\pm 0.5759905  i, \pm 0.817456 i $\\
\hline\noalign{\medskip}
$50$ & $\pm 0.57661177 i,\pm 0.81482 i$ &$\pm 0.570577 i,\pm 0.823433 i$&$\pm 0.572415 i,\pm 0.82188 i $&$\pm 0.5759905  i, \pm 0.817456 i $\\
\hline\noalign{\medskip}
$100$ & $\pm 0.57614517 i,\pm 0.816797  i$&$\pm 0.574613 i, \pm  0.818975 i$&$\pm 0.575084 i, \pm 0.818575 i $&$\pm 0.5759905 i, \pm 0.817456 i$\\ 
\hline\noalign{\medskip}
$ 400 $ & $\pm 0.5760001 i, \pm 0.817415 i$&$\pm 0.5759039 i, \pm 0.817552 i$&$\pm 0.5759336 i,\pm 0.817527 i $&$\pm 0.57599 i,\pm 0.817456 i $\\
\hline\noalign{\medskip}
$800$ & $\pm 0.575992904 i, \pm 0.817446 i $&$\pm 0.575969i,\pm 0.81748 i $&$\pm 0.575976 i, \pm 0.817474 i  $&$\pm 0.5759905i, 0.817456 i$\\
\hline\noalign{\medskip}
$ 1600$ &$\pm 0.57599109 i,\pm 0.817454 i$&$\pm 0.575985 i, \pm 0.817462 i$&$\pm 0.575987 i, \pm 0.817461 i$&$\pm 0.5759905 i, \pm 0.817456 i$\\
\hline\noalign{\medskip}
$3200$ & $\pm 0.57599064 i, \pm 0.817456 i$&$\pm 0.575989 i, \pm 0.817458 i$&$\pm 0.57599 i, \pm 0.817457 i $&$\pm 0.5759905  i, \pm 0.817456 i $\\
\hline\noalign{\medskip}
$ 6400$ &$\pm 0.57599053 i, \pm 0.817456 i$&$\pm 0.57599 i, \pm 0.817457 i$&$\pm 0.57599 i,\pm 0.817457 $&$\pm 0.5759905  i, \pm 0.817456 i $\\
 \hline\noalign{\medskip}
 $12800$ &$\pm 0.57599050 i, \pm 0.817456 i $&$\pm 0.57599 i, \pm 0.817456 i$&$\pm 0.57599 i, \pm 0.817456 i$&$\pm 0.5759905  i, \pm 0.817456 i $\\
 \hline\noalign{\medskip}
\end{tabular}
\end{center}
}

The following table compares the error of the period $T_1$ and $T_2$, see equation (\ref{t1t2}), when we compute them first using the roots of the polynomial $p_{2002}$, then,  using the roots of the polynomial $p_{2006}$ and finally using the polynomial $p_{Newton}=\lambda ^4+\lambda ^2-\frac{27}{4} (\mu -1) \mu$. We will be assuming that the exact values for $T_1$ and $T_2$ are those obtained by solving the system of equations that determine $L_4$ with a precision smaller than $10^{-30}$. We compute the error using Lemma (\ref{error}) assuming that the mass of the total system is the mass of the sun $M_0=1.988544* 10^{30}$. We point out that if we assume that the mass of the system is, let us say 5 times $M_0$, then all the error in the table would be 5 times bigger.

 {\tiny
\begin{center}
    \begin{tabular}{| c | c| c | c | }
 \hline\noalign{\medskip}
 $c_i$ & Error in seconds using $p_{2002} $ & Error in seconds using $p_{2006} $  &Error in seconds using $p_{Newton}$ \\
 \hline\noalign{\medskip}
$10$ &$0.00783531,\, 0.0124688 $&$0.00701609 ,\, 0.00871046 $&$0.00334932,\, 0.00167358$\\
\hline\noalign{\medskip}
$50$  &$0.049659, \, 0.0709552$&$0.0407816,\, 0.0491932$&$0.0153121, \, 0.00723655 $\\
\hline\noalign{\medskip}
$100$&$0.10076,  \,  0.143247$&$0.0822706, \, 0.0990907 $&$0.0305573 \, 0.014425$\\ 
\hline\noalign{\medskip}
$ 400 $ &$0.404933, \, 0.574751$&$0.329999,\, 0.397276 $&$0.122147,\, 0.0576405 $\\
\hline\noalign{\medskip}
$800$ &$0.810058,\, 1.14968$   &$0.660091, \, 0.794644$   &$0.244285, \, 0.115275$\\
\hline\noalign{\medskip}
$ 1600$ &$1.62021, \, 2.29945$  &$1.32023, \, 1.58933$  &$0.488565, \, 0.230547$\\
\hline\noalign{\medskip}
$3200$ &$3.24047, \, 4.59895$ &  $2.64048,\, 3.17869 $    &   $0.977129, \, 0.461093$\\
\hline\noalign{\medskip}
$ 6400$ &$6.48097, \, 9.19791$ &  $5.28098,\, 6.35739 $    &   $1.95426, \, 0.922185$\\
 \hline\noalign{\medskip}
  $12800$ &$12.9619, \, 18.3958$  &   $10.562, \, 12.7148$  &  $3.90851, \, 1.84437 $\\
 \hline\noalign{\medskip}
 \end{tabular}
\end{center}
}

The following table compute the distance between the Lagrangian point $L_4$,  computed with a precision smaller than $10^{-30}$ and  the coordinates given by the non relativistic coordinate $(\frac{1-2 \mu}{2},\frac{\sqrt{3}}{2})$. The table also contains the distance between  the Lagrangian point $L_4$,  computed with a precision smaller than $10^{-30}$, and the coordinates  $\left(\frac{1-2 \mu }{2} (\frac{5}{4 c^2}+1) ,\frac{\sqrt{3}}{2}\, (1-\frac{6 \mu ^2-6 \mu +5}{12 c^2})\right)$. Again we are using Lemma (\ref{error}) to compute these distances assuming that the mass of the total system is the mass of the sun $M_0=1.988544* 10^{30}$. We point out that if we assume that the mass of the system is, let us say 5 times $M_0$, then all the error in the table would be 5 times bigger.

 {\tiny
\begin{center}
    \begin{tabular}{| c | c| c | }
 \hline\noalign{\medskip}
 $c_i$ & Distance to $\left(\frac{1-2 \mu }{2} (\frac{5}{4 c^2}+1) ,\frac{\sqrt{3}}{2}\, (1-\frac{6 \mu ^2-6 \mu +5}{12 c^2})\right)$
 in meters  & Distance to $\left(\frac{1-2 \mu }{2}  ,\frac{\sqrt{3}}{2}\right)$ in meters\\
 \hline\noalign{\medskip}
$4$    &    $ 163.873 $ & $ 1162.89$\\
\hline\noalign{\medskip}
$10$    &    $ 23.095 $ & $ 1023.76$\\
\hline\noalign{\medskip}
$50$  &  $0.904045$& $1001.79 $\\
\hline\noalign{\medskip}
$100$&  $0.22586$& $1001.12 $\\ 
\hline\noalign{\medskip}
$ 400 $ &$0.0141133$& $1000.91 $\\
\hline\noalign{\medskip}
$800$ & $0.00352835$  & $ 1000.9$ \\
\hline\noalign{\medskip}
$ 1600$   &$0.000882579$  & $1000.9 $\\
\hline\noalign{\medskip}
$3200$ &  $0.000221304 $    & $ 1000.9$\\
\hline\noalign{\medskip}
$ 6400$  &  $0.0000589137 $ & $1000.9 $      \\
 \hline\noalign{\medskip}
  $12800$ &   $0.0000547835$  & $ 1000.9$\\
 \hline\noalign{\medskip}
 \end{tabular}
\end{center}
}

\section{Conclusion} 
\begin{itemize}
\item
When considering the relativistic case, the rounding error introduced by forgetting about the terms $\frac{1}{c^n}$ with $n\ge3$ is so big that, in all of the cases considered here, the characteristic polynomial given by the non relativistic case produced more accurate information that the characteristic polynomial introduced in the papers written in 2002, \cite{D}, and 2006, \cite{A}. 
\item
The conclusion found in \cite{A}, that states that $L_4$ is stable as long as $\mu<0.0384$ is not true.
\item
Even though the characteristic polynomials found before does not provide good approximations for $T_1$ and $T_2$, we have that the approximation for $L_4$ given by $\left(\frac{1-2 \mu }{2} (\frac{5}{4 c^2}+1) ,\frac{\sqrt{3}}{2}\, (1-\frac{6 \mu ^2-6 \mu +5}{12 c^2})\right)$ seems to be a good approximation.
\item
The Poincare-Miranda Theorem provides an useful tool to find the Lagrangian points with any given desired precision. 

\end{itemize}



\end{document}